\documentclass{jetpl}
\twocolumn
\lat

\usepackage[pdfhighlight=/P,colorlinks=true,linkcolor=blue,citecolor=blue,urlcolor=blue]{hyperref}
\newcommand{\hrefo}[1]{\href{http://oeis.org/#1}{#1}}
\newcommand{\citeo}[1]{\cite[\hrefo{#1}]{OEIS}}

\usepackage{amsthm}
\theoremstyle{plain}
\newtheorem{theorem}{Theorem}
\newtheorem{statement}[theorem]{Statement}

\title{Volterra chain and Catalan numbers}

\rtitle{Volterra chain and Catalan numbers}

\sodtitle{Volterra chain and Catalan numbers}

\author{V.\,E.\:Adler$^+$\/\thanks{e-mail: adler@itp.ac.ru}, A.\,B.\:Shabat$^+$}

\rauthor{V.\,E.\:Adler, A.\,B.\:Shabat}

\sodauthor{Adler, Shabat}

\address{$^+$L.\,D.\:Landau Institute for Theoretical Physics RAS, 142432 Chernogolovka, Russia}

\dates{30 October 2018}{*}

\abstract{We consider the Cauchy problem for the Volterra chain with an initial condition equal to 0 in one node and 1 in the others. It is shown that this problem admits an exact solution in terms of the Bessel functions. The Taylor series arising here are related to the exponential generating function for Catalan numbers.}


\begin{document}
\maketitle

\section{Generating function}\label{s:det}

Let us consider the Volterra chain
\begin{equation}\label{ut}
 u'_n=u_n(u_{n+1}-u_{n-1}),
\end{equation}
where $u_n=u_n(t)$ and the prime denotes the derivative with respect to $t$. By differentiating equations, one can express the higher derivatives $u^{(k)}_n$ as polynomials in $u_{n-k},\dots,u_{n+k}$. This makes possible to compute derivatives of any order by the initial data $u_n(0)$ and construct a solution of the Cauchy problem for (\ref{ut}) in the form of Taylor series. If the initial data is bounded, then it is easy to obtain estimates for the coefficients of these series, from which it follows that the radius of convergence is greater than 0 \cite{Kulaev_Shabat_2018}. As a rule, this radius is not very large, and the calculation of derivatives is a rather laborious procedure. However, there is at least one interesting case in which the coefficients of the Taylor series can be found explicitly and, thus, an exact solution of the chain can be obtained. Continuing the study started in \cite{Kulaev_Shabat_2018}, we consider special step-like initial conditions on the half-line $n\ge0$:
\begin{equation}\label{u011}
 u_0(0)=0,\quad u_1(0)=u_2(0)=\dots=1.
\end{equation}
In this case, a connection with several results from combinatorics is found: it turns out that the series under consideration are expressed in terms of the exponential generating function for Catalan numbers. In turn, this function is expressed through the Bessel function of the imaginary argument \citeo{A000108}. As a result, we obtain a solution  of the problem (\ref{ut}), (\ref{u011}) in terms of the Wronskians of this function and its derivatives. In combinatorics, this corresponds to the Hankel transform \cite{Layman_2001} which maps the Catalan numbers into unital sequence. It should be noted that the connection of Catalan numbers with integrable systems has been repeatedly noted, for example, in the theory of the dispersionless Toda chain \cite{Kodama_Pierce_2009}.

Together with the chain (\ref{ut}) it is convenient to consider the modified chain
\begin{equation}\label{vt}
 v'_n=v^2_n(v_{n+1}-v_{n-1})
\end{equation}
with initial data
\begin{equation}\label{v011}
 v_0(0)=0,\quad v_1(0)=v_2(0)=\dots=1.
\end{equation}
It is easy to verify that the solutions to this problem and the problem (\ref{ut}), (\ref{u011}) are connected by the substitution (discrete Miura transform)
\[
 u_n=v_nv_{n+1},\quad n=0,1,2,\dots
\]
The zero initial condition at $n=0$ implies that $u_0(t)=v_0(t)=0$ for all $t$ and the equations of both chains are naturally restricted to half-line $n>0$. Functions $u_1$ and $v_1$ are related by equation
\begin{equation}\label{u1v1}
 u_1=\frac{v'_1}{v_1}.
\end{equation}
A direct computation of several higher derivatives $v^{(n)}_1(0)$, $n=0,1,2,\dots$, by use of equations (\ref{vt}), (\ref{v011}), yields the Catalan numbers \citeo{A000108}
\[
 1,\,1,\,2,\,5,\,14,\,42,\,132,\,429,\,1430,\,4862,\,16796,\dots
\]
In similar way, for the derivatives $u^{(n)}_1(0)$, the use of equations (\ref{ut}), (\ref{u011}) or equation (\ref{u1v1}) brings to the sequence \citeo{A302197}
\[
 1,\,1,\,1,\,0,\,-4,\,-10,\,15,\,210,\,504,\,-3528,\,-34440,\dots
\]
which corresponds to the Hurwitz logarithm of the Catalan numbers. To substantiate these observations, we will need Wronskian representation of the solution of the chains \cite{Leznov_1980}.

\begin{statement}\label{st:w}
Let $v_0=0$ and $v_1=f(t)$, then solutions of the chains (\ref{vt}) and (\ref{ut}) at $n\ge0$ are of the form
\begin{equation}\label{uvw}
 v_n=\frac{w_{n-3}w_{n-1}}{w^2_{n-2}},~~ 
 u_n=\frac{w_{n-3}w_n}{w_{n-2}w_{n-1}},~ n=0,1,\dots,
\end{equation}
where $w_{-3}=0$, $w_{-2}=w_{-1}=1$ and, at $k\ge0$, 
\begin{gather*} 
 w_{2k}=\left|\begin{matrix}
   f       & f'        & \dots  & f^{(k)}  \\
   f'      & f''       & \dots  & f^{(k+1)}\\
   \vdots  & \vdots    & \ddots &\vdots    \\
   f^{(k)} & f^{(k+1)} & \dots  & f^{(2k)}
 \end{matrix}\right|,\\
 w_{2k+1}=\left|\begin{matrix}
   f'       & f''       & \dots  & f^{(k+1)}  \\
   f''      & f''       & \dots  & f^{(k+2)}\\
   \vdots   & \vdots    & \ddots &\vdots    \\
   f^{(k+1)}& f^{(k+2)} & \dots  & f^{(2k+1)}
 \end{matrix}\right|.
\end{gather*} 
\end{statement}

Let us compare these formulas with the Hankel transform for Catalan numbers \cite{Layman_2001}. Recall, that in addition to the explicit formula
\[
 c_n=\frac{(2n)!}{n!(n+1)!}=\frac{1}{n+1}\binom{2n}{n},
\]
Catalan numbers can be determined by any of the following recurrence relations:
\begin{gather}
\label{cn.1}
 c_{n+1}=\frac{4n+2}{n+2}c_n,~~ c_0=1,~~ n=0,1,\dots \\
\label{cn.2}
 c_{n+1}=c_0c_n+\dots+c_nc_0,~~ c_0=1,~~ n=0,1,\dots 
\end{gather}
For an arbitrary sequence $a_n$, the Hankel transform is the sequence of principal minors of an infinite matrix with element $(i,j)$ equal to $a_{i+j-1}$. One of the many remarkable properties of sequence $c_n$ is that it goes into the sequence of 1 under this transformation, and the same is true for the sequence with $c_0$ dropped \cite{Layman_2001}. 

\begin{statement}\label{st:Hankel}
The identities hold, at $k\ge0$:
\begin{gather*} 
 H_{2k}=\left|\begin{matrix}
   c_0    & c_1     & \dots  & c_k     \\
   c_1    & c_2     & \dots  & c_{k+1} \\
   \vdots & \vdots  & \ddots &\vdots   \\
   c_k    & c_{k+1} & \dots  & c_{2k}
 \end{matrix}\right|=1,\\
 H_{2k+1}=\left|\begin{matrix}
   c_1    & c_2     & \dots  & c_{k+1} \\
   c_2    & c_3     & \dots  & c_{k+2} \\
   \vdots & \vdots  & \ddots &\vdots   \\
   c_{k+1}& c_{k+2} & \dots  & c_{2k+1}
 \end{matrix}\right|=1.
\end{gather*} 
\end{statement}
\begin{proof}[Proof] 
Let us make use of elementary transformations of determinants taking the relations (\ref{cn.2}) into account. First, we consider $H_{2k}$ and apply the following transform on the rows $A_i=(c_i,c_{i+1},\dots,c_{i+k})$:
\[
 A_i\to A_i-c_{i-1}A_0-\dots-c_0A_{i-1},\quad i=1,\dots,k.
\]
According to (\ref{cn.2}), after this the 0-th column turns into $(1,0,\dots,0)^t$ and the determinant reduce to a determinant of size $k\times k$, with element $(i,j)$ equal to
\[
 c_{i+j}-c_{i-1}c_j-\dots-c_0c_{i+j-1} = c_0c_{i+j-1}+\dots+c_{j-1}c_i,
\]
that is,
\[
 \left|\begin{matrix}
  c_0c_1 & c_0c_2+c^2_1  & \dots & c_0c_k+\dots+c_{k-1}c_1 \\
  c_0c_2 & c_0c_3+c_1c_2 & \dots & c_0c_{k+1}+\dots+c_{k-1}c_2 \\
  \vdots & \vdots &  &\vdots \\
  c_0c_k & c_0c_{k+1}+c_1c_k & \dots & c_0c_{2k-1}+\dots+c_{k-1}c_k
 \end{matrix}\right|.
\]
Taking into account that $c_0=1$ and applying transformations over columns, we arrive at the determinant $H_{2k-1}$. Thus, we prove that $H_{2k}=H_{2k-1}$. Further on, we consider the determinant $H_{2k+1}$. After the transformation of rows
\[
 A_i\to A_i-c_{i-2}A_1-\dots-c_0A_{i-1},\quad i=2,\dots,k+1,
\]
we obtain the determinant of size $(k+1)\times(k+1)$, with element $(i,j)$ equal to (at $i>1$)
\[
 c_{i+j-1}-c_{i-2}c_j-\dots-c_0c_{i+j-2}=c_0c_{i+j-2}+\dots+c_{j-1}c_{i-1}.
\]
We replace, by (\ref{cn.2}), the elements of the first row as $c_{j+1}=c_0c_j+\dots+c_jc_0$ and arrive at
\[
 \left|\begin{matrix}
  c^2_0  & c_0c_1+c_1c_0     & \dots & c_0c_k+\dots+c_kc_0\\
  c_0c_1 & c_0c_2+c^2_1      & \dots & c_0c_{k+1}+\dots+c_kc_1 \\
  c_0c_2 & c_0c_3+c_1c_2     & \dots & c_0c_{k+2}+\dots+c_kc_2 \\
  \vdots & \vdots            &       &\vdots \\
  c_0c_k & c_0c_{k+1}+c_1c_k & \dots & c_0c_{2k}+\dots+c^2_k
 \end{matrix}\right|.
\]
The transformations over columns prove that $H_{2k+1}=H_{2k}$.
\end{proof}

We also need an expression for the exponential generating function of the sequence $c_n$ in terms of the modified Bessel function $I_0(z)$ (or the Bessel function of an imaginary argument), which is defined as a solution to the Cauchy problem
\begin{equation}\label{I0}
 zI''_0+I'_0(z)-zI_0(z)=0,~~ I_0(0)=1,~~ I'_0(0)=0.
\end{equation}

\begin{statement}[{\cite[A000108, K.\,A.\,Penson]{OEIS}}]\label{st:fI0}
The exponential generating function for Catalan numbers is
\begin{multline}\label{fI0}
 f(t)=c_0+c_1t+\dots+\frac{c_n}{n!}t^n+\dots\\ = e^{2t}(I_0(2t)-I'_0(2t)).\qquad
\end{multline}
\end{statement}
\begin{proof}[Proof]
Under the above change, the equation (\ref{I0}) becomes
\begin{equation}\label{f''}
 tf''(t)+2(1-2t)f'(t)-2f(t)=0,~~ f(0)=1.
\end{equation}
Replacing $f(t)$ by its Taylor expansion, we obtain the recurrence relations (\ref{cn.1}), as required.
\end{proof}

Now the statement that the derivatives of $v_1$ at $t=0$ coincide with the Catalan numbers is obtained as a simple consequence. Finally, we get the following result.

\begin{theorem}\label{th:sol}
The solutions of the Cauchy problems (\ref{ut}), (\ref{u011}) and (\ref{vt}), (\ref{v011}) are given by equations (\ref{uvw}) with determinants $w_n$ constructed from function $f(t)=e^{2t}(I_0(2t)-I'_0(2t))$.
\end{theorem}
\begin{proof}[Proof]
We find, by substitution $t=0$ in (\ref{uvw}), that $w_n(0)=1$ for all $n$. These relations are the same as relations $H_n=1$ which define the Catalan numbers uniquely. Therefore, function $v_1=f$ possesses the property that $f^{(n)}(0)=c_n$ and we only need to apply the Statement \ref{st:fI0}.
\end{proof}

\section{Solution asymptotics}\label{s:sol}

The function $f(t)$ is entire and has exponential growth at $t\to+\infty$. In contrast, function
\[
 u_1(t)=\frac{f'(t)}{f(t)}=\frac{I_0(2t)-I_2(2t)}{I_0(2t)-I_1(2t)}
\] 
has a finite limit at $t\to+\infty$ and pole singularities in the complex domain. Representation (\ref{uvw}) can be used for plotting of solution at $t$ greater than the convergence radius of the corresponding Taylor series. The figures \ref{fig:u1-u6},\,\ref{fig:t2} show that at $t\to+\infty$ the variables $u_n(t)$ with odd numbers grow, tending to a constant, and the variables with even numbers tend to 0; at $t\to-\infty$ all varables tend to 0. 

\begin{figure}[b!]
\centerline{\includegraphics[width=80mm]{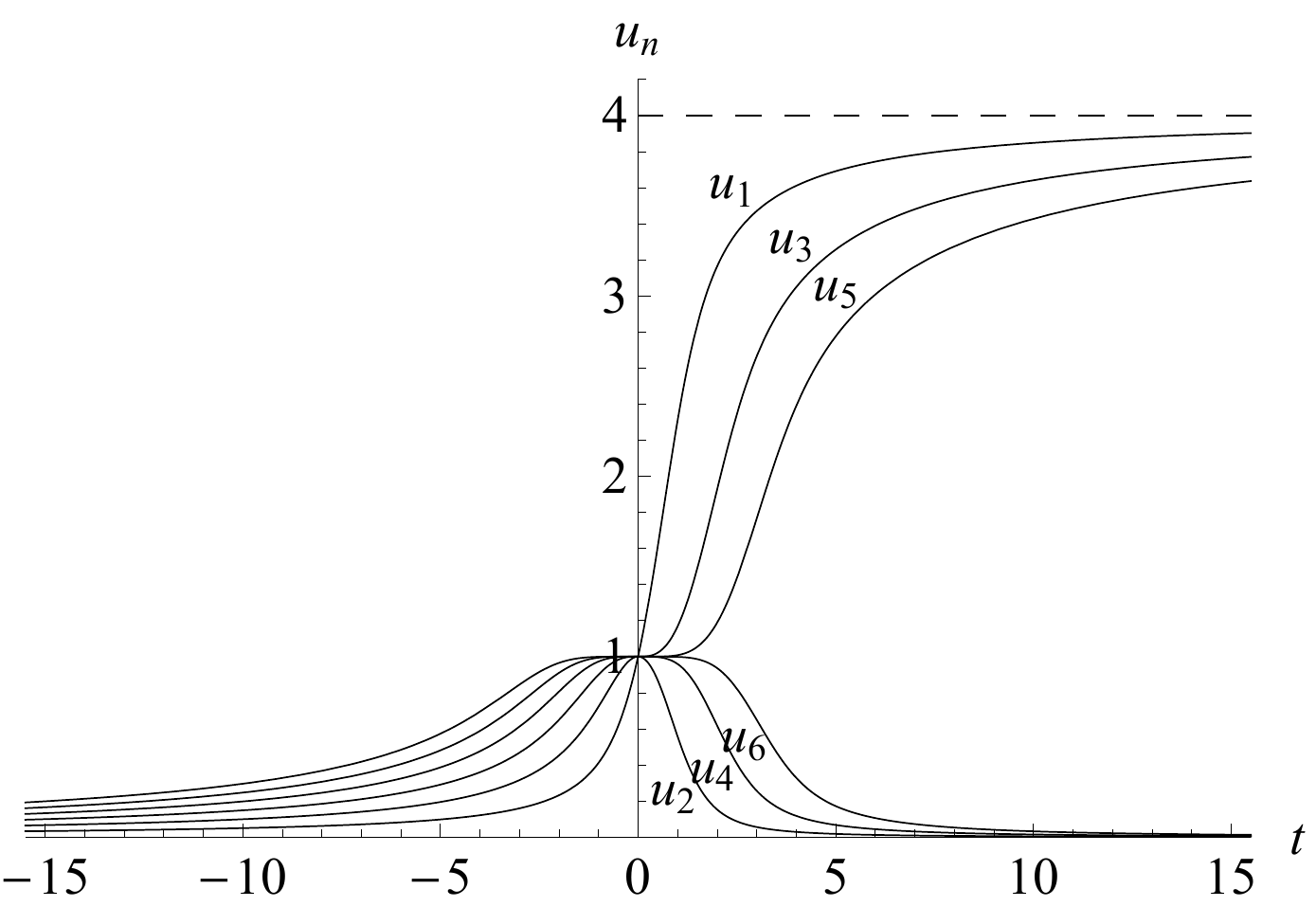}}
\caption{Functions $u_n(t)$}
\label{fig:u1-u6}
\end{figure}

\begin{figure}[b!]
\centerline{\includegraphics[width=80mm]{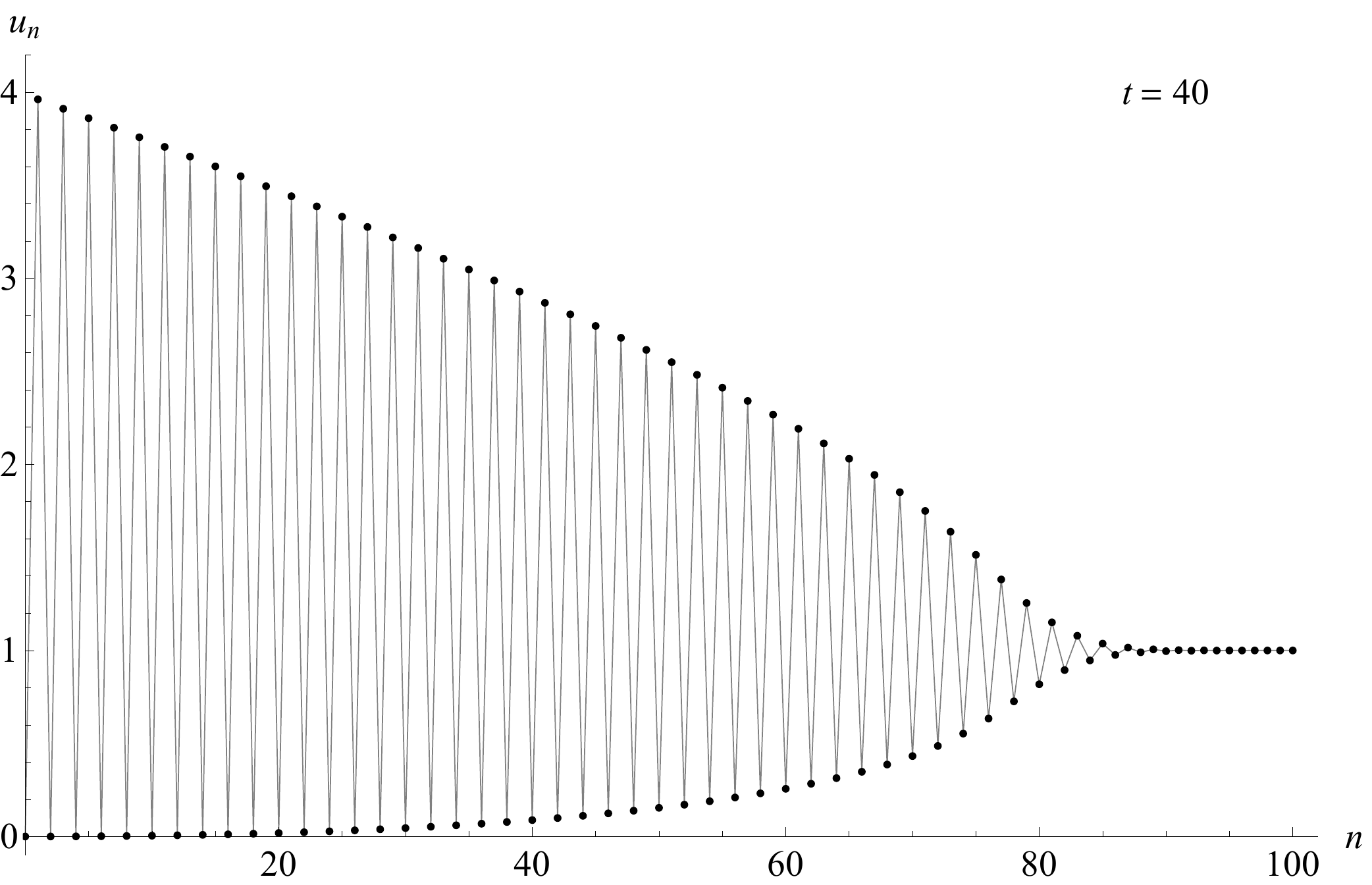}}
\caption{Plot of solution at $t=40$}
\label{fig:t2}
\end{figure}

To derive asymptotic formulas, we use the fact that the function $I_0(z)$ is even and its asymptotic behavior at $z\to+\infty$ is \cite{AS} 
\[
 I_0(z)=\frac{e^z}{\sqrt{2\pi z}}\Bigl(1+\frac{1}{8z}+\frac{9}{16z^2}+\dots\Bigr),
\]
where the coefficients of the series are calculated by substituting into equation (\ref{I0}). By passing to function $v_1(t)=f(t)$, we obtain the exponential growth for positive $t$ and power decrease for negative $t$:
\[
 f(t)=\left\{
  \begin{aligned} 
  & \frac{e^{4t}}{8\sqrt{\pi t^3}}\Bigl(1+\frac{3}{16t}+\frac{45}{512t^2}+\dots\Bigr),
  && t\to+\infty,\\
  & \frac{1}{\sqrt{-\pi t}}\Bigl(1+\frac{1}{16t}-\frac{3}{512t^2}+\dots\Bigr),
  && t\to-\infty.
 \end{aligned}\right.
\]
In both cases, the coefficients of the series in parentheses can be found by substituting into the equation (\ref{f''}). From here it is not difficult to obtain the asymptotic expansions for $u_1(t)$ as well, but an easier way is to transform (\ref{f''}) to Riccati equation
\[
 u'_1+u^2_1-\Bigl(4-\frac{2}{t}\Bigr)u_1-\frac{2}{t}=0.
\] 
Substitution $u_1=a_0+a_1t^{-1}+a_2t^{-2}+\dots$ implies that the free term can take values $a_0=4$, $a_0=0$, the other coefficients are calculated uniquely. Comparison with the formulas for $f$ shows that these two solutions correspond to asymptotics at $t\to\pm\infty$:
\[
 u_1(t)=\left\{
  \begin{aligned} 
  & 4-\frac{3}{2t}-\frac{3}{16t^2}-\frac{9}{64t^3}-\dots,&& t\to+\infty,\\
  &  -\frac{1}{2t}-\frac{1}{16t^2}+\frac{1}{64t^3}+\dots,&& t\to-\infty.
 \end{aligned}\right.
\]
Further on, it follows, directly from the chain equations (\ref{ut}), that for $t\to+\infty$ the variables $u_n$ have asymptotics
\[
 u_{2j}=\frac{\alpha_j}{t^2}+\dots,~~ u_{2j+1}=4-\frac{\beta_j}{t}+\dots,~~ j=0,1,2,\dots
\]
with coefficients obeying recurrence relations
\[
 \alpha_{j+1}-\alpha_j=\frac{\beta_j}{4},~~ \beta_{j+1}-\beta_j=2,~~ 
 \alpha_0=0,~~ \beta_0=\frac{3}{2}.
\]
Similarly, for $t\to-\infty$ we have $u_n=\dfrac{\gamma_n}{t}+\dots$ and
\[
 \gamma_0=0,~~ \gamma_1=-\frac{1}{2},~~
 \gamma_{n+2}-\gamma_n=-1,~~ n=0,1,2,\dots
\]
Solving of these equations proves the following statement. 

\begin{statement}\label{st:uinf}
The asymptotic expansions hold
\[
\left\{\begin{aligned}
 & u_{2j-1}(t)=4-\frac{4j-1}{2t}+O(t^{-2}),\\
 & u_{2j}(t)=\frac{j(2j+1)}{8t^2}+O(t^{-3}),
\end{aligned}\right.\quad t\to+\infty,
\]\[
 u_n(t)=-\frac{n}{2t}+O(t^{-2}),~~ t\to-\infty.
\]
\end{statement}

\section{Further generalizations}\label{s:generalizations}

The studied solution is not exceptional. By experiments with different chains and initial conditions, one can find other examples where the coefficients of Taylor series allow a closed expression or a combinatorial interpretation. We hope that study of this relationship may be fruitful. In conclusion, we give two more examples of this kind.\smallskip

1) Let us consider instead of (\ref{vt}) the modified Narita--Ito--Bogoyavlensky chain of order $m$ \cite{Bogoyavlensky_1991}
\[
 v'_n=v^2_n(v_{n+m}\cdots v_{n+1}-v_{n-1}\cdots v_{n-m}),
\]
with initial data $v_0(0)=0$, $v_n(0)=1$ at $n>0$. As before, in this case the equations are restricted on the half-line $n>0$. Direct calculation of derivatives $v_1$ at $t=0$ brings to the sequences from \cite{OEIS}
\begin{align*}
 m=2:&~ 1,\,1,\,3,\,12,\,55,\,273,\,1428,\,7752,\dots &&\hrefo{A001764}\\
 m=3:&~ 1,\,1,\,4,\,22,\,140,\,969,\,7084,\dots       &&\hrefo{A002293}\\
 m=4:&~ 1,\,1,\,5,\,35,\,285,\,2530,\,23751,\dots     &&\hrefo{A002294}\\
 m=5:&~ 1,\,1,\,6,\,51,\,506,\,5481,\,62832,\dots     &&\hrefo{A002295}
\end{align*}
known as the Pfaff--Fuss--Catalan numbers \cite{Graham_Knuth_Patashnik_1990}, with general formula
\[
 C^{(m+1)}_n= \frac{1}{mn+1}\binom{(m+1)n}{n},\quad n=0,1,2,\dots
\]
($m=1$ corresponds to the Catalan numbers). However, the proof that $v^{(n)}_1(0)=C^{(m+1)}_n$ is not known at the moment, for $m>1$.\smallskip

2) Consider the modified Volterra chain again, but with a linearly growing initial condition:
\[
 v'_n=v^2_n(v_{n+1}-v_{n-1}),\quad v_n(0)=n,\quad n\ge0.
\]
The computation of derivatives $v^{(n)}_1(0)=a_n$ by virtue of the chain brings to the sequence \citeo{A000182}
\[
 1,\,2,\,16,\,272,\,7936,\,353792,\,22368256,\dots
\]
with the property that if we interleave it with zeros, then the exponential generating function will be 
\[
 \tan x = a_0x+a_1\frac{x^3}{3!}+\dots+a_n\frac{x^{2n+1}}{(2n+1)!}+\dots
\]
However, the original sequence grows too fast and the Taylor series $v_1(t)=a_0+a_1t+\dots+a_n\dfrac{t^n}{n!}+\dots$ has the zero convergence radius, that is, the solution is formal in this example.\smallskip 

The research of V.E.A. was supported by RFBR grant 16-01-00289a, the research of A.B.S. was supported by RScF grant 15-11-20007.



\begin{thebibliography}{99}
\bibitem{Kulaev_Shabat_2018} R.\,Ch.\:Kulaev, A.\,B.\:Shabat, Conservation laws for the Volterra chain with an initial step-like condition, accepted in Ufa Math. J., (2018). 

\bibitem{OEIS} N.\,J.\,A.\:Sloane, The On-Line Encyclopedia of Integer Sequences, published electronically at \url{http://oeis.org}, 2010.
    
\bibitem{Layman_2001} J.\,W.\:Layman, J. of Integer Sequences {\bf 4}, Article \href{http://emis.ams.org/journals/JIS/VOL4/LAYMAN/hankel.pdf}{01.1.5} (2001).
 
\bibitem{Kodama_Pierce_2009} Y.\:Kodama, V.\,U.\:Pierce, Comm. Math. Phys. {\bf 292:2}, \href{https://doi.org/10.1007/s00220-009-0894-1}{529} (2009).

\bibitem{Leznov_1980} A.\;N.\:Leznov, Theor. Math. Phys. {\bf 42:3}, \href{https://doi.org/10.1007/BF01018624}{225} (1980).

\bibitem{Bogoyavlensky_1991} O.\,I.\:Bogoyavlensky, Russ. Math. Surveys {\bf 46:3}, \href{http://dx.doi.org/10.1070/RM1991v046n03ABEH002801}{1} (1991).

\bibitem{Graham_Knuth_Patashnik_1990} R.\,L.\:Graham, D.\,E.\:Knuth, O.\:Patashnik, Concrete Mathematics. Addison-Wesley, Reading, MA, 1990.

\bibitem{AS} M.\:Abramowitz, I.\,A.\:Stegun (eds.) Handbook of Mathematical Functions with Formulas. New York: Dover, 1972. 
\end{thebibliography}
\end{document}